\documentclass[11pt]{article}
\usepackage{amsfonts,bbm}
\usepackage{epsfig,amsmath,amssymb,graphics,mathrsfs,amsthm}
\usepackage{eepic,bm,epsfig,color,subfigure}
\def\bra#1{\mathinner{\langle{#1}|}}
\def\ket#1{\mathinner{|{#1}\rangle}}

\def\Bra#1{\left<#1\right|}

\def\inp#1#2{\Bra{#1}#2\rangle}

\newcommand{\pj}[1]{\ket{#1}\bra{#1}}
\newcommand{\inpr}[3]{\mathinner{\langle #1|#2|#3\rangle}} 
\newcommand{\beq}{\begin{equation}}
\newcommand{\eeq}{\end{equation}}
\newcommand {\cali}[1]{{\mathcal #1}}

\newcommand{\tr}{{\tt Tr}} 
 
\newcommand{\conj}[1]{\overline{#1}}
 
\newcommand{\pmat}{\begin{pmatrix}} 
\newcommand{\emat}{\end{pmatrix}} 
\newcommand{\complex}{\mathbb{C}}
\newcommand{\real}{\mathbb{R}}
 
\newcommand{\tensor}{\otimes}

\renewcommand{\Re}{\operatorname{Re}}
\renewcommand{\Im}{\operatorname{Im}}
\newcommand{\rnk}{\operatorname{rank}}
{}{}
\newtheorem{thm}{Theorem}{}{}
\newtheorem{lem}{Lemma}{}{}
\newtheorem{cor}{Corollary}{}{}
{}{}

\newcommand{\rank}{\operatorname{rank}}

\newcommand{\commentout}[1]{}
\newcommand{\be}{\begin{enumerate}}
\newcommand{\ee}{\end{enumerate}}
\newcommand{\bi}{\begin{itemize}}
\newcommand{\ei}{\end{itemize}}
\newcommand{\bpmat}{\begin{pmatrix}}
\newcommand{\epmat}{\end{pmatrix}}
\newcommand{\F}{\mathscr{F}} 
\newcommand{\G}{\mathscr{G}} 
\newcommand{\B}{\cali{B}}
\newcommand{\la}{\langle} 
\newcommand{\ra}{\rangle} 
\newcommand{\assign}[1]{\la\B_{#1},S_{#1}\ra}

\begin{document}
\title{Consistent assignment of quantum probabilities} 
\author{Manas K. Patra\\Department of Computer Science, University of York\\York YO10 5GH, UK \\and \\Laboratoire d'Information Quantique,
Universite Libre de Bruxelles\\ Campus Plaine, Bruxelles 1050, Belgium
 \and Ron van~der Meyden\\School of Computer Science \& Engineering, University of New South Wales\\ Sydney 2052, Australia} 
\date{}
\maketitle
\begin{abstract}
We pose and solve a problem concerning consistent assignment of quantum probabilities to a set of bases associated with maximal projective measurements. We show that our solution is optimal. We also consider some consequences of the main theorem in the paper in conjunction with Gleason's theorem. Some potential applications to state tomography and probabilistic quantum secret-sharing scheme are discussed. 
\end{abstract}
\section{Assignment of quantum probabilities} \label{sec:defnCons}
Consider a finite quantum system $\mathscr{S}$ of dimension $n$. 
Define a {\em quantum probability assignment} (QPA) for $\mathscr{S}$ to be 
a set $\F = \{ \la \B_1, S_1\ra, \ldots, \la \B_m, S_m\ra\}$ 
consisting of $m$ pairs $\la \B_k, S_k\ra$ where 
\be 
\item 
each $\B_k = \la P^k_1, \ldots, P^k_n\ra$ is a sequence of $n$ 1-dimensional projection operators $P^k_i = \pj{\alpha^k_i}$
corresponding to an orthonormal basis $\ket{\alpha^k_1}, \ldots , \ket{\alpha^k_n}$, and 
\item
each $S_k = \la p^k_1, p^k_2, \dotsc, p^k_n \ra$ is a sequence of $n$ real numbers such that  
$p^k_i\geq 0$  and $ \sum_i p^k_i= 1$. 
\ee
Clearly, the sequences $S_k$ can be considered as distributions on a probability space with $n$ elements (atomic events).

Suppose now it is claimed that these probabilities correspond to projective measurements in these $m$ bases. 
Thus, the claim is that there is a (mixed) quantum state $\rho$ (a positive semidefinite matrix with trace 1) such that
\beq \label{eq:basicForm2}
\tr(\rho P^k_i) = p^k_i \text{ for all } 1\leq k \leq m, 1\leq i \leq n
\eeq
If such a state exists we say that the 
collection $\F$ is a {\em consistent quantum probability assignment}. 
How do we verify this claim? Potentially $m$ could be infinite. A related question is the following \\

\noindent
{\bf Problem.} What is the {\em smallest} number 
$r$ such that if all the probability assignments $\G\subseteq \F$ with 
$|\G| \leq r$ ($|\G|$ denotes cardinality of $\G$) are consistent, then the whole collection $\F$ is consistent?\\

\noindent
The number 
depends on the assignment $\F$. 
So we ask what is the largest value of $r(\F)$ as we vary $\F$. This value depends
only on the dimension $n$. Call it the {\em consistency number} $r_n$. Equivalently, 
$r_n$ is the smallest number such that for every
collection of quantum probability assignments $\F$, if every subset of $\F$ of size $r_n$ is consistent then $\F$ is consistent. 
This  problem first came up in our attempt at formal axiomatization of a logic of 
finite-dimensional quantum systems \cite{MP1}. The problem we address in this paper is the calculation of the consistency number. 

We mention that a related problem was posed by Bell and Kochen and Specker \cite{Bell,KS} in their analysis of non-contextual hidden variable theories. Projection operators of rank 1 are the quantum analogues of classical propositional variables. The latter can take two values, say 0 and 1, corresponding to {\bf true} and {\bf false}. The question then is whether there is a consistent probability assignment to the whole collection such that the $p^k_i$ 
take value 1 or 0. Clearly quantum mechanics prohibits this in case of sets of incompatible projection operators. The question posed by these authors is whether such assignment is possible for an arbitrary collection of bases in some hidden variable theory. 

The paper is structured as follows. 
In Section~\ref{sec:upper} and \ref{sec:optimal}, we prove that $r_2=4$ and 
$r_n=n^2-n+1$ for $n\geq 3$. 
We first prove a weaker result that if every subset consisting of $r_n+1=n^2-n+2$ probability assignments is consistent then so is the 
whole set $\F$. The proof in the case of $r_n$ is surprisingly much harder. There is one exceptional case where one has to investigate in detail the structure of the bases themselves (unlike the case for $r_n+1$). But we get extra dividends in gaining information about the measurement bases. In fact, there is essentially one case where we need the maximal value $n^2-n+1$. We use this in Section \ref{sec:optimal} to construct examples to demonstrate that the number $r_n$ is indeed optimal. 
More precisely, if $r'<r_n$ there exist quantum probability assignments $\F$  such that all 
subcollections of $\F$ of size $r'$ are consistent but $\F$ itself is not consistent. 
The tricky part in the construction of these examples is that we have to ensure that positive definite matrices with trace 1 (states) exist which satisfy the assignments corresponding to {\em every} subcollection of size $r'$. We conclude the section  with a simple application  of the main result in conjunction with Gleason's famous theorem
\cite{Gleason}. In Section \ref{sec:app} we outline some  applications of our constructions to (probabilistic) secret sharing and state tomography.  
Section~\ref{sec:concl} makes some concluding remarks.

\section{Solution to the problem} \label{sec:upper} 

This section is devoted to 
establishing  upper bounds for the consistency number. 
We assume first that the number of assignments $m$ is finite, 
but we will show later that this assumption can be dropped. More precisely, we prove: 

\begin{thm}\label{thm:basicCons}
Let $H$ be the Hilbert space of an $n$-dimensional quantum system. Define a quantum probability assignment  $\F = \{ \la \B_1, S_1\ra, \ldots, \la \B_m, S_m\ra\}$ 
consisting of $m$ pairs $\la \B_k, S_k\ra$ where 
\be
\item [{\em i.}]
 $\B_k = \la P^k_1, \ldots, P^k_n\ra$ is a sequence of $n$ projection operators $P^k_i = \pj{\alpha^k_i}$
corresponding to an orthonormal basis $\ket{\alpha^k_1}, \ldots , \ket{\alpha^k_n}$, and 
\item[{\em ii.}]
 $S_k = \la p^k_1, p^k_2, \dotsc, p^k_n \ra$ is a sequence of $n$ real numbers such that  
$p^k_i\geq 0$  and $ \sum_i p^k_i= 1$. 
\ee

Then $\F$ is a consistent quantum probability assignment to an $n$-dimensional quantum system if and only if every subset $\G \subseteq \F$ of size at most 
\beq
R_n
= \begin{cases} 4 & \text{ if } n=2\\
n^2-n+1 & \text{ if } n>2\end{cases}
\eeq
is consistent. 
That is, $r_n \leq R_n$. 
\end{thm}

Note that arbitrary $k$-subsets $\F$ may be separately consistent. However, the proof below will produce a state in which the whole of $\F$ is a consistent assignment provided the conditions of the theorem are satisfied. Observe also that the problem of assignment arises because of the unitary transformation connecting two bases. 
We will later see that these bounds are optimal in the sense that for a positive integer $r'< R_n$ 
there are bases and corresponding probability assignments such that every subset of size $r'$ is consistent but the whole set is inconsistent. 

We will start with a relatively simpler result. 

\begin{lem}\label{lem:Consist_simple}
With the notation as above,  
$\F= \{\la \B_1,S_1\ra, \dotsc, \la \B_m,S_m\ra\}$  is a 
consistent probability assignment to an $n$-dimensional quantum system if and only if every subset $\G \subseteq  \F$
of size at most $n^2-n+2$ is consistent. 
\end{lem}
\begin{proof}
Let 
$\F= \{\la \B_1,S_1\ra, \dotsc, \la \B_m,S_m\ra\}$ 
where $\B_i$ are bases and $S_i$ the corresponding probability assignments. 
We may assume that $m> n^2-n+2$. 
Let $I$ denote the unit matrix of order $n$.
Now for each $1\leq k\leq m$ we have a set of consistency equations
\beq \label{eq:basicForm3}
\begin{split}
\tr(\rho I)& =\tr(\rho)=1 \\
\tr(\rho P^k_i) & = p^k_i ~\text{for}~ 1\leq i \leq n-1 \\
\end{split}
\eeq
in the {\em real} variables $x_{jk}$ and $y_{jk}$ where $\rho_{jk}=x_{jk}+iy_{jk}$, the $(jk)^{\text{th}}$ entry of the quantum state $\rho$. 
Note that we omit the equation $\tr(\rho P^k_n)  = p^k_n$ for the case $i=n$ since this
follows from others using the facts that $\sum_i P^k_i= I$ and $\sum_ip^k_i=1$. 
The hermiticity of $\rho$ implies that there are $n^2$ independent parameters. Any one set of equations corresponding to a single basis assignment fixes the trace condition. We call the set of equations corresponding to the assignment to the $k$th basis the $k$th block, $B_k$. So we have $n^2$ independent real variables and any block, say the first, has $n$ independent equations. That is, the rank of the coefficient matrix of the first block is $n$. But after the normalization ($\tr(\rho)=1$) is fixed, each block from the second block onwards can contribute at most $n-1$ to the overall rank. Given an arbitrary set ${\cal A}$ of assignments let $\rnk({\cal A})$ denote the rank of the coefficient matrix of the corresponding set of linear equations \eqref{eq:basicForm3}. Since $\rho$ has $n^2$ real variables $\rnk({\cal A})\leq n^2$. Now consider a subset $\G\subset \F$ of cardinality $l\geq n^2-n+2$. Call such subsets $l$-subsets of $\F$. We will show that if every $l$-subset is consistent then any $(l+1)$-subset is consistent. Then induction will complete the proof. Hence, we have to show that given any $(l+1)$-subset $\mathscr{K}= \{\la \B_i,S_i\ra :1\leq i \leq l+1\}$ it is consistent if every $l$-subset of $\mathscr{K}$ is consistent. 
\def\K{\mathscr{K}}

By assumption $\G_i=\K-\{\la\B_i,S_i\ra\}$ is consistent for all $1\leq i\leq l+1$. Now consider the sets $\G_{ij}=\G_i- \{\la\B_j,S_j\ra\}$. Suppose $\rank(\G_{ij})<n^2$. Since $|\G_{ij}| \geq n^2-n+1$ this implies that at least one of the blocks say $B_k$ in $\G_{ij}$ is dependent on the rest in the sense that $\rank(\G_{ij}-\{\la\B_k,S_k\ra\})=\rank(\G_{ij})$.  Suppose this is not the case. Then we must have 
\[ n = \rank(\{\la\B_1,S_1\ra\}) < \rank(\la\B_1,S_1\ra,\la\B_2,S_2\ra\})<\dotsc < \rank(\G_{ij})<n^2\]
But this is impossible since $|\G_{ij}|\geq n^2-n+1$. We have shown that there is some $k$ such that $\la\B_k,S_k\ra\in \G_{ij}$  and $\rank(\G_{ij}-\{\la\B_k,S_k\ra\})=\rank(\G_{ij})$. Let $\G'_{ij}=\G_{ij}-\la\B_k,S_k\ra$. 
Then every equation in the block $B_k$ can be written as a linear combination of equations of $\G'_{ij}$. Hence consistency of $\G_{ij}$ implies that {\em any} solution to $\G'_{ij}$ is a solution for $\G_{ij}$. The cardinality of $\G'_{ij}\bigcup \{\{\la B_i,S_i\ra,\la B_i,S_i\ra\}$ is $l$ and by hypothesis it is consistent. But from what we have shown above a solution $\rho$ to this system is also a solution to $\G'_{ij}\bigcup \{\{\la B_i,S_i\ra,\la B_j,S_j\ra, \la B_k,S_k\ra\}=\K$. The proof is complete in this case. 

Next consider the second alternative, $\rank(\G_{ij})=n^2$. This is the maximal rank since there are $n^2$ variables. Hence, any solution to the system $\G_{ij}$ is unique. This implies that $\G_i$ and $\G_j$ have the same unique solution. But then so does $\K$ and the proof is complete. 
\end{proof}

The following lemma provides one of the key elements in improving the 
bound so as to obtain our main result. 

\begin{lem}\label{lem:mainThm}
Let $\cali{B}=\{\epsilon_1, \dotsc, \epsilon_n\}$ and $\cali{B}'=\{\beta_1, \dotsc, \beta_n\}$ be two bases in an $n$-dimensional Hilbert space. Let $B$ and $B'$ be, respectively, 
the associated matrices of coefficients for the systems 
of real equations 
\beq
\begin{split}
&\tr(\pj{\epsilon_i}\rho) = p_i,~\text{and}~ \tr(\pj{\beta_i}\rho) = q_i, ~(\text{both for}~ i=1,\dotsc, n) \\
\end{split}
\eeq 
where $\sum_ip_i=\sum_iq_i=1$ and 
where $\rho$ is positive semidefinite matrix and the entries of $\rho$ are treated as unknown variables. Suppose the rank of the combined system $B\text{ and }B'$ is $n+1$, exactly 1 higher than the rank of $B$ (or $B'$). Then two of the vectors from $\cali{B}'$, say, $\beta_1$ and $\beta_2$ lie in the plane determined by two vectors from $\cali{B}$. Assuming (without loss of generality) that the later is spanned by $\epsilon_1\text{ and }\epsilon_2$ we also have $\inp{\epsilon_i}{\beta_j}\neq 0\text{ for } i,j\in \{1,2\}$. Further, the set $\{\beta_j:3\leq j\leq n\}$ is a permutation of $\{\epsilon_j:3\leq j\leq n\}$. 
\end{lem}
\begin{proof}
Writing the matrix elements of $\rho$ in the $\cali{B}$ basis we observe that the equations corresponding to $B$ determine the diagonal elements $\rho_{ii}$ along with the constraint $\sum_i\rho_{ii}=1$. Now let 
\(
\beta_j = \sum_i c_{ij} \epsilon_i 
\). Then the equations in $B'$ are equivalent to the following. 
\[
\begin{split}
& \sum_j|c_{jl}|^2x_{jj}+\sum_{j< k} (\conj{c}_{jl}c_{kl}+ c_{jl}\conj{c}_{kl})x_{jk}+ i(\conj{c}_{jl}c_{kl}  - c_{jl}\conj{c}_{kl})y_{jk}= q_l \text{ where }\\
& \rho_{jk} = \inpr{\epsilon_j}{\rho}{\epsilon_k} = x_{jk}+iy_{jk}
\end{split}
\]
Using $B$ we can eliminate the diagonal terms from these equations (corresponding to $j=k$). Then the hypothesis that the rank of the combined system $B$ and $B'$ is exactly one more than that of $B$ alone implies that 
\beq\label{eq:rank1}
\conj{c}_{jl}c_{kl} = \alpha_{lm} \conj{c}_{jm}c_{km}\; \forall j\neq k
\eeq
Here $\alpha_{lm}$ are {\em real} constants depending on the basis vectors $\beta_l$ and $\beta_m$ or equivalently the $l$th and $m$th row vectors of $B'$. Now, among the $\beta_i$'s there must be at least two vectors which have at least two nonzero coefficients when expressed in the vectors of $\cali{B}$. Otherwise, $\cali{B}$ and $\cali{B}'$ are the same basis apart from permutation. We may take these two vector to be $\beta_1$ and $\beta_2$. We may also assume without loss of generality that $c_{11}$ and $c_{21}$ are both nonzero. Suppose for $\beta_2$ the coefficients $c_{j2}$ and $c_{k2}$ are nonzero. Then, addition of either of the equations corresponding to $\beta_1$ or $\beta_2$ to the system corresponding to $B$ increases rank by 1. The equation \ref{eq:rank1} then shows that $\alpha_{12}$ {\em and } $\alpha_{21}(=\alpha_{12}^{-1})$ are both nonzero. Hence, $c_{21}$ and $c_{22}$ must be nonzero.  
Then equation \eqref{eq:rank1} implies that $c_{12}$ and $c_{22}$ are both nonzero. Using equation \ref{eq:rank1} again we infer that there is a nonzero constant $\gamma$ such that $c_{j2}=\gamma c_{j1}$ for $j= 3,\dotsc, n$. If a third basis vector, say $\beta_3$, had two nonzero coefficients a similar type of argument would show the existence of a nonzero constant $\gamma'$ such that $c_{j3}=\gamma'c_{j1},\;j=2,3,\dotsc,n$. But then $\beta_1,\beta_2\text{ and }\beta_3$ will not be linearly independent. Hence, all $\beta_i,\; i=3,\dotsc, n$ must have exactly one nonzero coefficient $c_{ki}$. Since they are unit vectors this must 1 (apart from an inconsequential factor of modulus 1). That is, $\beta_i=\epsilon_{j_i}$. Since the $\beta_i$'s constitute an orthogonal basis $\epsilon_{j_i}$ cannot be any of the $\epsilon_k$ whose coefficients $c_{k1}$ (and $c_{k2}$) is nonzero. In particular, $j_i\neq 1\text{ or }2$. But there are $n-2$ such $\beta_i$'s. Hence, the set $\{\beta_3, \beta_4,\dotsc, \beta_n\}$ must be a permutation of $\{\epsilon_3, \epsilon_4, \dotsc, \epsilon_n\}$ and $\{\beta_1,\beta_2\}$ form an orthonormal basis in the 2-dimensional space spanned by $\epsilon_1\text{ and }\epsilon_2$. Moreover, from the fact that $c_{11},c_{21},c_{12}\text{ and } c_{22}$ are all nonzero we conclude that the $\inp{\epsilon_i}{\beta_j}\neq 0\text{ for } i,j\in \{1,2\}$. The proof of the lemma is complete. 
\end{proof}

We now improve  Lemma~\ref{lem:Consist_simple} so as to yield a proof of Theorem~\ref{thm:basicCons}, 
by considering several special cases. 

\noindent
\begin{proof}
The implication from left to right
is trivial --- any subset of a consistent set is consistent. 
To prove the converse, we may assume $m>R_n$.  

The case $n=2$ is covered by Lemma \ref{lem:Consist_simple}. So we assume that $n\geq 3$ and  that any $R_n$-subset
of $\F$ is consistent. The problem can be restated as follows. Find a positive semidefinite matrix $\rho$ satisfying the set of linear equations \eqref{eq:basicForm3}. There are $n^2$ independent real variables corresponding to the real and imaginary parts of entries of $\rho$ and hence the real rank of system is $\leq n^2$. Any one of the blocks corresponding to one basis assignment has rank $n$. For each of the rest of blocks of equations corresponding to the pair $\la \B_k, S_k\ra$ we need to verify $n-1$ equations in each block since the last one is guaranteed by the condition $\sum_ip^k_i=1$. Again using 
Lemma \ref{lem:Consist_simple} it suffices to show that every 
$(R_n+1)$-subset 
of $\F$ is consistent. Assume otherwise: there is a set 
$\G=\{\la \B_1, S_1\ra, \dotsc, \la \B_{R_n+1}, S_{R_n+1}\ra\} $ 
such that $\G$ is inconsistent while every proper subset of $\G$ is consistent. 

Recall that the $\rank(\G)$ of a quantum probability assignment is defined to be the rank of the coefficient matrix of the consistency equations \eqref{eq:basicForm3} corresponding to $\G$. 
The set $\G'= \{\la \B_1, S_1\ra, \dotsc, \la \B_{R_n}, S_{R_n}\ra \}$ is consistent by hypothesis. 
Suppose that for some pair  $\la \B_i, S_i\ra, \la \B_{j}, S_{j}\ra\in \G'$, 
the rank of the combined system satisfies 
\beq\label{eq:condRank}
\rank(\{\B_i, S_i\ra, \la \B_{j}, S_{j}\ra\}) \geq  n+2 ~.
\eeq
Consider the chain of inequalities 
\[
\begin{split}
 & \rank(\{\la\B_i,S_i\ra\}) = n  
<n+2 \leq \rank(\la\B_i,S_i\ra,\la\B_j,S_j\ra\})\leq \\ 
& \rank(\{\la\B_i,S_i\ra,\la\B_j,S_j\ra,\la\B_1,S_1\ra\})\leq \dotsc \leq  \rank(\G')\leq n^2 \\
\end{split}
\]
where in each term we add exactly one new basis assignment to the previous set. Since there are $R_n=n^2-n+1$ terms  at least one of the relations $\leq$ must be an equality. Hence, the equations corresponding to at least one of the assignments, say, $\la \B_k,S_k\ra$ must be dependent on the rest. Thus, 
{\em every} solution of $\G'- \{\la \B_k,S_k\ra\}$ must be a solution of $\G'$. Now, $\G_k=\G- \{\la \B_k,S_k\ra\}$ is consistent being a set of cardinality $R_n$. From what we have just seen any solution to $\G_k$ is also a solution for $\la \B_k,S_k\ra$. Hence $\G$ is consistent, a contradiction. We observe that in general we arrive at the same conclusion if any of the relations in the chain 
\beq\label{eq:chain}
\rank(\{\la\B_{i_1},S_{i_1}\ra\})\leq \rank(\{\la\B_{i_1},S_{i_1}\ra,\la\B_{i_2},S_{i_2}\ra\})\leq \dotsc \leq \rank{\G'} 
\eeq
where $(i_1,i_2,\dotsc, i_{R_n})$ is a permutation of $(1,2, \dotsc, R_n)$ is an equality. Note also that if $\rank(\G')<n^2$ 
then an argument similar to the above would show that some assignment $\{\B_k,S_k\}$ must be dependent on the rest {\em without} requiring the condition in \eqref{eq:condRank}. We can then show that $\G$ is consistent. Hence, to complete the proof we have to consider only the following case. 

\vspace{.5 cm} 
\noindent
{\bf Condition 1}. 
For every distinct pair $(i,j)$  in $1\ldots R_n+1$, 
\beq \label{eq:rank_1}
\rank(\{\la\B_i,S_i\ra,\la\B_j,S_j\ra\})= \rank(\{\la\B_i,S_i\ra\})+1=n+1
\eeq
and for every $R_n$-subset $\G' = \{\la\B_{i_1},S_{i_1}\ra,\la\B_{i_2},S_{i_2}\ra, \ldots \la\B_{i_{R_n}},S_{i_{R_n}}\ra \}$ of $\G$,
in the chain 
\beq
\rnk(\{\la\B_{i_1},S_{i_1}\ra\})< \rnk(\{\la\B_{i_1},S_{i_1}\ra,\la\B_{i_2},S_{i_2}\ra\})<\cdots <\rnk(\G')
\eeq
the rank increases exactly by 1 at each step. \\ 

This condition implies that the set of consistency equations corresponding to every $R_n$-subset $\G'$ of $\G$ have rank $n^2$. Let us write the basis vectors of the basis $\B_k$ as 
\[ \B_k=\{ \beta^k_1,\beta^k_2, \dotsc, \beta^k_n\}\]
Without loss of generality we may take the basis $\B_1$ to be any fixed basis, say, the ``computational'' basis $\{\epsilon_1,\dotsc,\epsilon_n\}$.\footnote{Here the computational basis denotes the standard basis in $\complex^n$ where $\epsilon_i$ is the vector with 0's everywhere except the $i$th entry which is 1.}
Then the relation \eqref{eq:rank_1} and lemma \ref{lem:mainThm} imply that 
for the pair of distinct indices $(1,j)$  
from $1..R_n+1$ there exist pairs of distinct  
indices $(r,s)$ and $(j_r, j_s)$ from $1..n$ such that  
\beq \label{eq:rank1Cond}
\beta^j_{r}= a^j_{rr}\epsilon_{j_r}+a^j_{rs}\epsilon_{j_s} \text{ and } \beta^j_{s}= a^j_{sr} \epsilon_{j_r}+a^j_{ss}\epsilon_{j_s} 
\eeq
where $a^j_{rr},a^j_{rs}$ etc.\ are non-zero complex numbers. Moreover, the set $\B_j-\{\beta^j_{r}, \beta^j_{s}\}$ is a permutation of $\B_1-\{\epsilon_{j_r}, \epsilon_{j_s}\}$. 
As before we write $\rho_{ij}=\inpr{\epsilon_i}{\rho}{\epsilon_j}=x_{ij}+iy_{ij}$ in the computational basis $\B_1$. Then the assignment $\assign{1}$ fixes the diagonal entries $\rho_{ii}$ 
and the addition of equations corresponding to $\la \B_j, S_j\ra$ to those of  $\la \B_1, S_1\ra$ yields exactly two relations 
\[\Re(\conj{a}^j_{rr}a^j_{rs}\rho_{j_rj_s})=p^j_{r}\text{ and }\Re(\conj{a}^j_{sr}a^j_{ss}\rho_{j_rj_s})=p^j_{s}\]
Observe that $\inp{\beta_{r}}{\beta_{s}}=0$ implies that $a^j_{rr}\conj{a}^j_{sr}=-a^j_{rs}\conj{a}^j_{ss}$. From this it follows that the left sides of the two equations are dependent. Since they are consistent the right sides have similar dependence. That is, the addition of an arbitrary assignment $\assign{j}\in \G'$ to $\assign{1}$ yields exactly {\em one} additional equation
involving $x_{j_rj_s}$ and $y_{j_rj_s}$. Call it the $j$-incremental equation $\cali{E}_j$. 

It follows that the equations for $\G$ are equivalent to the set of equations 
for $\assign{1}$, which fix the diagonal, together with the equations $\cali{E}_2, \ldots, \cali{E}_{R_n+1}$, 
each of which gives exactly one equation on the real and complex parts of some off-diagonal entry. 
(To get this from the above, note that the equations for $\{\assign{1}\} \cup X\cup \{\assign{j}\}$ 
are equivalent to those for $\{\assign{1}\} \cup X\cup \{\assign{1}, \assign{j}\}$, which are 
equivalent to $\{\assign{1}\} \cup X\cup \{\assign{1}, \cali{E}_j\}$, i.e., 
$\{\assign{1}\} \cup X\cup \{\cali{E}_j\}$, and apply induction.)
Moreover, since  $\rank(\G)=n^2$, for each off diagonal entry  there are exactly two 
indices $j,k$  in $2\ldots R_n$ such that $\cali{E}_j$  and $\cali{E}_k$  concern that entry. 
Taking the entry to be that which is constrained by $\cali{E}_{R_n+1}$, we find that 
$\cali{E}_j$  and $\cali{E}_k$ and $\cali{E}_{R_n+1}$ concern the same entry. 
But $\G'$ is equivalent to $\assign{1}$ with $\cali{E}_2, \ldots, \cali{E}_{R_n}$,
which is consistent, and $\G$ is equivalent to adding to this $\cali{E}_{R_n+1}$, producing an 
inconsistency. The only way this is possible is if  $\cali{E}_j$  and $\cali{E}_k$ and $\cali{E}_{R_n+1}$
are inconsistent. Since $R_n > 4$ for $n>3$, we may find an $R_n$-subset $\G''$ of $\G$ 
that contains $\assign{1},\assign{j},\assign{k}$ and $\assign{R_n+1}$. But the
equations for this set then imply  the inconsistent set $\{ \cali{E}_j,  \cali{E}_k, \cali{E}_{R_n+1}\}$, 
so $\G''$ is inconsistent. This contradicts the assumption that every $R_n$-subset of $\G$ is 
consistent. 
\end{proof}

Observe that the proof of  theorem \ref{thm:basicCons} actually provides an algorithm for verifying consistency of a given set $\F$ of assignments. Thus we start with a an arbitrary  $\G\subset \F$ of cardinality $R_n$. If it is inconsistent stop, and declare that $\F$ is inconsistent. 
Otherwise find $\rank(\G)=k$, say. If $k=n^2$ then we have hermitian matrix $\rho$ which satisfies the assignments in $\G$. The problem then is to check whether the solution is positive definite. This can be done by using the algorithm for Cholesky decomposition \cite{Stewart}. If the algorithm succeeds then, of course, $\rho$ is positive (semi)definite. If it fails then we stop, declaring the system inconsistent. Next, suppose $k< n^2$. Choose a subset of $\G'$ of $\G$ that has maximal rank and let $|\G'|=s$. Add, $R_n-s$ new assignments and continue the process. If we exhaust $\F$ before attaining full rank then we have a system of equations with rank $r<n^2$. Then the hermitian matrix $\rho$ is determined with $n^2-r$ free parameters. To determine whether there is a positive definite solution in the corresponding parameter space we may again follow a Cholesky decomposition type of algorithm but now symbolic.  In the next section, we present bases in 2 and 3 dimension, inspired by the proof of the theorem, to demonstrate that the $R_n$ is optimal for dimension $n$. That is, we have determined the consistency number. We conclude this section with an application of the Tukey's lemma (equivalent to axiom of choice) \cite{Kelley} that yields the following corollary. We use the notations explained before. 
\begin{cor}
Let 
$\F=\{\assign{j}~|~ j\in J \}$
be a collection of quantum probability assignments to bases $\cali{B}_j$ 
in an $n$ dimensional space, indexed by a set $J$ of arbitrary cardinality. Then $\F$ is consistent if and only if every finite subset of cardinality $n^2-n+1$ is consistent. 
\end{cor}

\section{Optimality}\label{sec:optimal}
In this section we complete the proof that the consistency number $R_n$ given in Theorem~\ref{thm:basicCons} is optimal. 
We show that there exist quantum probability assignments $\F=\{(\cali{B}_k, S_k)~|~ k=1,2, \dotsc,\; m\geq R_n\}$
such that every $R_n-1$-subset of $\F$ is consistent (hence also every $r$-subset of $\F$ for $r<R_n$ is consistent) 
but $\F$ is not. For example, in 2 dimensions the bound is 4. So we have to show that there exist {\em inconsistent} probability assignments to 4 bases such that any three of them is consistent. For 3 dimensions the bound is 7. We obviously expect the construction of bases and corresponding probability assignments to be much harder in 3 and higher dimensions. We construct these examples in this section. 

\subsection{Dimension=2}
Let the dimension $n=2$ and define 1-dimensional projections
\[ 
\begin{split}
&P_1^1= \bpmat 1 & 0 \\ 0 & 0\epmat, \quad
P_1^2= \frac{1}{2}
\bpmat 1 & -1 \\ -1 & 1 \epmat,
\quad P_1^3= \frac{1}{2}\bpmat 1 & -i \\ i & 1 \epmat \\
&\text{ and } P_1^4=   \frac{1}{5} \bpmat 4 &  (6+8i)/5 \\ (6-8i)/5  & 1 \epmat 
\end{split}
\]

Let the four bases (rather the projections corresponding to the basis vectors) be given by
\[
\cali{B}_j= \{P^j_1, I-P^j_1\},\; j=1,\dotsc,4 
\]
where $I$ is the 2-dimensional unit matrix. Then the following simultaneous probability assignment is unsatisfiable although any {\em three} of them is satisfiable. 
\[ 
p(P_1^1) =\frac{1}{2}, \; p(P_1^2)=\frac{5}{12} ,\;p(P^3_1)= \frac{3}{8}\text{ and } p(P_1^4)=\frac{9}{16}
\]
In two dimensions we need specify the probability corresponding to only one of the outcomes in a projective measurement. It is routine to verify that this assignment is satisfiable for any three bases but not for all four. The two dimensional case is exceptional as evidenced by the failure of theorems of Kochen and Specker and Gleason\cite{Gleason}.  
\subsection{ Dimension 3 and higher}\label{sec:dim3}
In three dimensions the consistency number according to Theorem \ref{thm:basicCons} is $r_3\leq 7$. We will construct a projective measurement system consisting of 7 bases and corresponding probability assignments such that {\em any} subset with 6 bases is consistent but the whole set is not satisfiable showing that the consistency number is indeed 7. The bases will be specified by the corresponding projectors and since each basis is a complete set (the projectors add to $I$, the identity matrix in 3 dimensions) we need specify only two of them. Let 
\[ \cali{B}_i= \{P^i_1, P^i_2, P^i_3\}, \; i=1,\dotsc, 7\]
\[
\begin{split}
&P^1_1 =\begin{pmatrix} 1 & 0 & 0 \\ 0 & 0 & 0 \\ 0 & 0 & 0 \end{pmatrix},P^1_2 = \begin{pmatrix} 0 & 0 & 0 \\ 0 & 1 & 0 \\ 0 & 0 & 0 \end{pmatrix} \quad P^2_1 = \frac{1}{2}\begin{pmatrix} 1 & 1 & 0 \\ 1 & 1 & 0 \\ 0 & 0 & 0 \end{pmatrix}, \\
& P^2_2 = \frac{1}{2}\begin{pmatrix} 1 & -1 & 0 \\ -1 & 1 & 0 \\ 0 & 0 & 0 \end{pmatrix}\quad 
P^3_1 = \frac{1}{2}\begin{pmatrix} 1 & 0 & 1 \\ 0 & 0 & 0 \\ 1 & 0 & 1\end{pmatrix}, P^3_2 = \frac{1}{2}\begin{pmatrix} 1 & 0 & -1 \\ 0 & 0 & 0 \\ -1 & 0 & 1\end{pmatrix} \\
& P^4_1 = \begin{pmatrix} 0 & 0 & 0 \\ 0 & 1 & 1 \\ 0 & 1 & 1 \end{pmatrix}, 
P^4_2 = \begin{pmatrix} 0 & 0 & 0 \\ 0 & 1 & -1 \\ 0 & -1 & 1 \end{pmatrix} \quad 
P^5_1 = \frac{1}{2}\begin{pmatrix} 1 & -i & 0 \\ i & 1 & 0 \\ 0 & 0 & 0 \end{pmatrix}, \\
& P^5_2 = \frac{1}{2}\begin{pmatrix} 1 & i & 0 \\ -i & 1 & 0 \\ 0 & 0 & 0 \end{pmatrix}\quad P^6_1 = \frac{1}{2}\begin{pmatrix} 1 & 0 & i \\ 0 & 0 & 0 \\ -i & 0 & 1\end{pmatrix}, P^6_2 = \frac{1}{2}\begin{pmatrix} 1 & 0 & -i \\ 0 & 0 & 0 \\ i & 0 & 1\end{pmatrix} \\
&  P^7_1 = \frac{1}{2}\begin{pmatrix} 0 & 0 & 0 \\ 0 & 1 & -i \\ 0 & i& 1\end{pmatrix} \text{ and } P^7_2 = \frac{1}{2}\begin{pmatrix} 0 & 0 & 0 \\ 0 & 1 & i \\ 0 & -i & 1\end{pmatrix} \\
\end{split}
\]
The seven bases given above are particularly simple. We now add another base $\cali{B}'=\{Q_1,Q_2\}$ with more complicated structure. 
\[
Q_1 = \begin{pmatrix} \frac{1}{3} & \frac{e^{7\pi i/12}}{\sqrt{6}}  & \frac{e^{\pi i/3}}{3\sqrt{2}} \\ \frac{e^{-7\pi i/12}}{\sqrt{6}} & \frac{1}{2} & \frac{e^{-\pi i/4}}{2\sqrt{3}} \\
\frac{e^{-\pi i/3}}{3\sqrt{2}} & \frac{e^{\pi i/4}}{2\sqrt{3}} & \frac{1}{6} \end{pmatrix} \text{ and } 
Q_2 = \frac{6}{11}\begin{pmatrix} \frac{1}{2} & \frac{e^{-3\pi i/4}}{\sqrt{6}}  & \frac{e^{-\pi i/3}}{\sqrt{2}} \\ \frac{e^{3\pi i/4}}{\sqrt{6}} & \frac{1}{3} & \frac{e^{5\pi i/12}}{\sqrt{3}} \\
\frac{e^{\pi i/3}}{\sqrt{2}} & \frac{e^{-5\pi i/12}}{\sqrt{3}} & 1 \end{pmatrix}
\]
Now suppose we have the probability assignments
\beq \label{eq:assign1}
v(P^i_j)= p^i_j \text{ and }v(Q_j)=q_j,\; i=1, \dotsc, 7 \text{ and } j=1,2.
\eeq
If there exists a density matrix 
\beq \label{eq:density1}
\rho = \begin{pmatrix} a_1 & z_3 & \bar{z}_2 \\ \bar{z}_3 & a_2 & z_1 \\ z_2 & \bar{z}_1 & a_3 \end{pmatrix}, \quad a_i\geq 0 \text{ and } a_1+a_2+a_3=1
\eeq
with measurement probabilities $\tr(\rho P^i_j) = p^i_j$ then the following sets of equations determine the real parts of $z_i$. 
\begin{subequations} \label{eq:consProbRe}
\begin{align}
& a_1= p^1_1 \text{ and } a_2= p^1_2 \label{eq:Re1}\\
& \Re(z_3) = p^2_1- (a_1+a_2)/2= (a_1+a_2)/2-p^2_2 \label{eq:Re2}\\
&\Re(z_2) = p^3_1- (a_1+a_3)/2= (a_1+a_3)/2-p^3_2 \label{eq:Re3}\\
& \Re(z_1) = p^4_1- (a_2+a_3)/2= (a_2+a_3)/2-p^4_2 \label{eq:Re4}
\end{align}
\end{subequations}
Similarly, we have the imaginary parts determined by the bases \(\cali{B}_5, \cali{B}_6, \text{ and } \cali{B}_7\). 
\begin{subequations} \label{eq:consProbIm}
\begin{align}
& \Im(z_3) = (a_1+a_2)/2 - p^5_1= p^5_2-(a_1+a_2)/2 \label{eq:Im1}\\
&\Im(z_2) = (a_1+a_2)/2 - p^6_1= p^6_2-(a_1+a_2)/2\label{eq:Im2}\\
& \Im(z_1) = (a_1+a_2)/2 - p^7_1= p^7_2-(a_1+a_2)/2 \label{eq:Im3}
\end{align}
\end{subequations}
Note that the 
addition of each of the blocks for the bases $\cali{B}_2 , \dotsc, \cali{B}_7$ to the block for basis $\cali{B}_1$
increases the rank by exactly 1. 
Finally, for the base $\cali{B}'=\{Q_1, Q_2\}$ we have 
\beq \label{eq:consProbGen}
\begin{split}
&q_1= \frac{a_1}{3}+\frac{a_2}{2}+\frac{a_3}{6}+\frac{2\Re(z_3e^{\frac{-i7\pi}{12}})}{\sqrt{6}}+ 
\frac{2\Re(\bar{z}_2e^{\frac{-i\pi}{3}})}{3\sqrt{2}}+\frac{\Re(z_1e^{\frac{i\pi}{4}})}{\sqrt{3}}\\
& q_2= \frac{6}{11}\left(\frac{a_1}{2}+\frac{a_2}{3}+a_3+\frac{2\Re(z_3e^{\frac{i3\pi}{4}})}{\sqrt{6}}+ 
\frac{2\Re(\bar{z}_2e^{\frac{i\pi}{3}})}{\sqrt{2}}+\frac{2\Re(z_1e^{\frac{-i5\pi}{12}})}{\sqrt{3}}\right)\\
\end{split}
\eeq
The choice of bases yielding these equations follows a pattern. Each block 
for basis $\cali{B}_i,\; i> 1$ determines exactly one unknown. For example, 
 $\cali{B}_3$
 fixes $\Re(z_3)$ (see \eqref{eq:Re2}). Moreover, it is clear that the sets of equations (\ref{eq:consProbRe}) and (\ref{eq:consProbIm}) are {\em independent}. So, if two distinct subsets consisting of 6 equations each are satisfied by some density matrix then we have already found a unique solution. We infer that this set of 7 bases cannot used to show that $r_3=7$. Hence we have an additional base  $\cali{B}'$. Let us then drop one of the bases, say $\cali{B}_7$, and replace it with $\cali{B}'$. Let ${\mathcal O }= \{\cali{B}_1,\dotsc, \cali{B}_6, \cali{B}'\}$. The last basis $\cali{B}'$ has the property that if we omit any other basis say $\cali{B}_6$ 
 from ${\mathcal O}$ then the remaining system of equations has maximal rank. That is, the resulting system of equations (\ref{eq:consProbRe}), (\ref{eq:consProbIm}$'$) (omitting \eqref{eq:Im3}) and (\ref{eq:consProbGen}) is {\em overdetermined}. We describe below a ``procedure'' for finding probability assignments to the bases in $\cali{O}$ such that any six of them are consistent but the whole set is not. 

First, a necessary and sufficient condition for $\rho$ given in (\ref{eq:density1}) to be density matrix is that the following relations are satisfied. 
\beq \label{eq:density2}
\begin{split}
&\det(\rho)= a_1a_2a_3 -\sum_{i=1}^3 a_i|z_i|^2 +2\Re(z_1z_2z_3)\geq 0,\\
&  a_1a_2-|z_1|^2\geq 0, ~~ a_3a_1-|z_2|^2\geq 0, ~~ a_1a_2 -|z_1|^2\geq 0 \text{ and } a_1,a_2\geq 0\\
\end{split}
\eeq
These conditions reflect the fact that for a hermitian matrix to be positive definite it is necessary and sufficient that all the principal minors have non-negative determinant. Now choose a density matrix $\rho$ such that the inequalities in \eqref{eq:density2} are strict. We also choose $a_1,a_2>0$ such that $a_1+a_2<1$ and substitute $a_3=1-a_1-a_2$. Now treat the $a_i$, $\Re(z_i)$ and $Im(z_i)$ as (real) variables. Then the left side of all the inequalities above are continuous functions of these variables. Here we consider $\rho$, parametrized by $\{a_1,a_2,\Re(z_i),\Im(z_i):i=1,2,3\}$, as a member of $\real^8$ ($\real$ is the field of reals). Hence there is an open neighborhood $N_1$ of $\rho$ such that these inequalities hold everywhere in $N_1$. Compute $\tr(\rho P^i_j)=p^i_j,\; i=1,\dotsc,6$, $\tr(\rho Q_1)=q_1$ and $\tr(\rho Q_2)=q_2'$ (note the ``prime''). From the set of bases $\cali{O}$ if we drop any base from the set $\{\cali{B}_1,\dotsc,\cali{B}_6\}$ then the remaining set of equations have maximal (real) rank 8. Let $A_i$ denote the invertible matrix of maximal rank consisting of coefficients of a set of {\em independent} equations corresponding to the omission of $\cali{B}_i,\;i=2, \dotsc, 6$. 

The invertibility of the matrices $A_i,\; i=1,\dotsc, 6$ implies that the images $A_i(N_1)$ are open sets and hence $\bigcap_i A_i(N_1)\equiv G$ is open. The point $\alpha'=(p^1_1,p^1_2, \dotsc, p^6_1,p^6_2, q_1,q'_2)^T$ is in $G$ as it is the image of $\rho$. We choose a point $q_2\neq q'_2$ such that if the last ``coordinate'' $q'_2$ in $\alpha'$ is replaced by $q_2$ then the resulting vector $\alpha\in G$. Since $G$ is open such a choice is always possible. This choice of $q_2$ makes the new assignments for  $\cali{O}$ inconsistent. But restricting this assignment to any six bases is still consistent. First, if we omit $\cali{B}'$ then clearly the assignment in rest of the bases is satisfiable by $\rho$ itself. If we omit any other base say $\cali{B}_i$ then a desired density matrix, say $\gamma_i$ is given by $A_i^{-1}(\alpha)$. We thus conclude that:  
\begin{quote} 
{\em The probability assignments to the seven bases $\cali{O}$ given above are such that any six of them is consistent with a quantum state (density matrix) but the whole set is not}. 
\end{quote} 
Next, it is clear that we can mimic the construction of the bases given above for 3 dimensions in any dimension $n>3$. Thus, define the bases $\cali{B}_0 \text{ and }\cali{B}_{ij},\, i>j$ as follows. Let $e_i=(0,\dotsc, 0,1,0\dotsc,0)^T$ ($i$th coordinate=1) be the standard basis (the computational basis). Then, $\cali{B}_0=\{P^0_i=e_ie_i^T: i=1, \dotsc, n-1\}$ is the basis consisting of projectors on the standard basis. They determine the diagonal elements of the state $\rho$. Next, let 
\beq\label{eq:basisRe}
\begin{split}
& \cali{B}_{ij} =\{ P^{ij}_k:1\leq i<j\leq n\text{ and } 1\leq k\leq n-1\},\; P^{ij}_i= \frac{1}{2} (e_i+e_j)(e_i+e_j)^T, \\ 
& P^{ij}_j= \frac{1}{2} (e_i-e_j)(e_i-e_j)^T 
\text{ and } P^{ij}_k = e_k e_k^T \; k\neq i,j \\
\end{split}
\eeq
These $n(n-1)/2$ bases determine the real parts of the off-diagonal elements of $\rho$. These correspond to \(\cali{B}_2, \cali{B}_3 \text{ and } \cali{B}_4\) in the 3-dim case above. We similarly define the $n(n-1)/2$ bases for the imaginary parts of off-diagonal elements of $\rho$. 
\beq \label{eq:basisIm}
\begin{split}
& \cali{B}'_{jk} =\{ P^{jk}_l:1\leq j<k\leq n\text{ and } 1\leq k\leq n-1\},\; P^{jk}_j= \frac{1}{2} (e_j+ie_k)(e_j+ie_k)^T, \\ 
& P^{jk}_k= \frac{1}{2} (e_j-ie_k)(e_j-ie_k)^T
\text{ and } P^{jk}_l = e_l e_l^T \; l\neq j,k \\
\end{split}
\eeq
As in the 3 dimensional case we replace one of the bases, say $\cali{B}'_{n-1,n}$, by a basis $Q$. Call the new system of bases $\cali{Z}$. The basis $Q$ is chosen so that the probability assignments to all the bases in $\G$ yield an over-determined system of equations. This can be achieved by ensuring that the rank the system corresponding to $\G-\cali{B}_{ij}\text{ and } \cali{Z}-\cali{B}'_{ij},\; \forall 1\leq i<j\leq n$ is maximal ($=n^2$). We can then use topological arguments similar to the 3 dimensional case to show that every subsystem of $\cali{Z}$ has a consistent solution but the full system $\cali{Z}$ consisting of $R_n=n^2-n+1$ bases does not. 
Thus, the
number $R_n$ in Theorem \ref{thm:basicCons} is the best possible. We can now state the following. 
\begin{thm}
The consistency number of an $n$-dimensional quantum system is  
$r_n = R_n$. 
\end{thm}
Combining this theorem with Gleason's theorem \cite{Gleason} we get the following theorem. First, recall some definitions needed to state Gleason's famous result. Let $v$ be a function on the set $\cali{P}(H)$ of projections on a finite-dimensional Hilbert space $H$ such that
\begin{subequations}\label{eq:frame}
\begin{align}
& 0\leq v(E) \leq 1,\; v(I)=1 \label{frame1}\\
&v(E+F)=v(E)+v(F) \text{ if } EF=0 \text{ {\em (orthogonal projections)}} \label{frame2}
\end{align}
\end{subequations}
Such a function is called frame function. 
\begin{thm}
Let 
${\mathscr Z}=\{\cali{U}_i~|~i\in \cali{I} \}$ where $\cali{I}$ is an indexing set  and 
$\cali{U}_i= \{E^i_1,\dotsc, E^i_{k_i}\}$ consists of orthogonal projections: $E^i_jE^i_l= \delta_{jl}$
on a Hilbert space of dimension $n\geq 3$.  
Suppose a  real-valued function $f$ on ${\mathscr Z}$  satisfies \eqref{frame1}. Further, assume that for  every set $S\subset {\mathscr Z}$ of cardinality $n^2-n+1$ 
there exists a {\em frame} function $v_S$ such that $f(E)=v_S(E)$ for all $e\in \cup S$. 
Then there exists a frame function $v$ such that $f=v$ on $\cup {\mathscr Z}$. 
\end{thm}
\begin{proof}
Gleason's theorem states that any frame function $w$ on a Hilbert space 
of dimension at least 3 
is induced by some density operator $\rho$: $w(E)=\tr(\rho E)$. This implies that for every subset $S$ of ${\mathscr Z}$ of cardinality $n^2-n+1$ the function $f$ defines a consistent quantum probability assignment in the sense explained in Section \ref{sec:defnCons}. However as the projections $E_i^j$ may not be 1-dimensional we cannot apply Theorem \ref{thm:basicCons} directly. In this case we adopt the following procedure. Let $E$ be a projection operator of rank $k>1$. We decompose $E=E_1+\dotsc+E_k$ where $E_i$ are 1-dimensional projection operators with $E_iE_j=\delta_{ij}$. We replace each member of $\cali{U}_i$ by the projectors in its 1-dimensional decomposition. Let the resulting set of {\em orthonormal} projectors be $\cali{U}'_i$. Let ${\mathscr Z}'=\{\cali{U}'_i\}$. It is easy to see that any $S'\subset {\mathscr Z}'$ has consistent probability assignment induced by the original assignment on $S\subset {\mathscr Z}$. Therefore, from Theorem \ref{thm:basicCons} there is a frame function $v$ (given by a density matrix) which yield the same probabilities. The additivity of the frame functions now implies that $f=v$ on $\cup {\mathscr Z}$. 
\end{proof}

\section{Applications} \label{sec:app}
In this section we consider some applications of the constructions in the previous section. First, we sketch a {\bf secret sharing scheme}  involving copies of 
entangled qubits. A version of the secret sharing problem \cite{Shamir} is as follows. A group of $k$ players are to share a secret (represented by a number). If any  subset of at least $r$ players pool their resources (their {\em shares}) the secret is revealed otherwise it is not. This is called a $(k,r)$ (threshold) secret sharing scheme. Our scheme is probabilistic and we only require that any set of $r$ players  can discover the secret with high probability whereas for less than $r$ players the probability is low. We sketch a scheme using entangled qubits. 

Suppose we prepare multiple copies of a composite quantum system of dimension $N$. To each of the $k$ players we send a number of copies of  the system (or some part thereof) along with instructions for specific measurements. If we arrange it so that the probability of 
reconstructing the set from any $r\leq k$ expectation values is high, but negligible for any subset of cardinality $r'< r$ measurements then we have a $(k,r)$ secret sharing scheme. We outline below such a scheme for a system of $n$ qubits. Let $N=2^n$, the dimension of the system. For simplicity we will only consider a $(k,k)$ secret sharing scheme. 
\be
\item
A large number of copies of a state randomly chosen from an initial set $S$ of $K_0$ states is prepared by the {\em dealer}. 
\item
The dealer provides the players 
with $m$ copies each from the original ensemble called their {\em shares}.  
\item
The players are also given instructions about their respective bases in which projective measurements are to be performed: the bases are from $\cali{B}_{ij}$ or $\cali{B}'_{ij}$ given in  \eqref{eq:basisRe} and \eqref{eq:basisIm}. There are $N^2-N$ such bases. Thus each player is actually measuring the real or imaginary part of an off-diagonal element of the density matrix. 
\item
We assume that the diagonal elements are conveyed to each player along with the measurement instruction. These may be used as group ``password'' and/or as a check for interference. 
\item
We assume that the set $S$ of states are so chosen that when all but one of the (real) parameters characterizing the off-diagonal elements of $\rho$ are fixed there still a large number $\geq K_1$ of possible states 
with different values of the remaining parameter. We omit the details of how this is done in 
in this sketch. It implies that $K_0= O(K_1n^2)$. The values of each parameter are separated by a distance $>\lambda$. 
\item
Suppose $\alpha$ is one of the parameters. We see from equations like \eqref{eq:basisRe} and \eqref{eq:basisIm} that the expectation values yield the probabilities and hence the value of $\alpha$. Using the Chernoff bound (classical) it is seen that if $\alpha_0$ is the correct value of $\alpha$ and $\bar{\alpha}$ the calculated value then 
\[ \text{probability } (|\alpha-\bar{\alpha}|> \lambda) <O(e^{-\lambda^2}/m) \]
Therefore if we are aware of the states in $S$ then with high probability we can determine the parameter $\alpha$. 
\item
If all the $N^2-N$ players combine their measurement results then (along with the information about diagonal elements) the particular state $\rho$ is determined with high probability. However, with just one player missing the probability drops to $< 1/K_1$. 
\ee
In this probabilistic protocol we assume that we have secure quantum channels with negligible errors. It is possible to devise a more elaborate scheme to accommodate insecure channels. Similarly, we can devise a general $(k,r)$ secret sharing scheme combining the quantum scheme with a classical one involving polynomial evaluation. 

The selection of the set $S$ is a bit more challenging. There are several possible approaches however. One is to start with a positive definite matrix and then keeping all the values of the parameter vary one of the parameters characterizing the off-diagonal elements. For example, let $\alpha_{ij}=\Re(\rho_{ij})$ and $\beta_{ij}=\Im(\rho_{ij}),\; i< j$ and suppose we want to vary $\beta_{12}$. Using topological arguments as before we know that there some neighborhood of $\beta_{12}$ such that for all values of the latter in that neighborhood we get a positive definite matrix (fixing the normalization is trivial). We divide the neighborhood into segments of appropriate size (fixing $\lambda$) and pick our values for $\beta_{12}$. Alternatively, we could use the fact that for any hermitian matrix $A$, $A^2$ is positive definite. We vary the parameters of $A$ to achieve our goal. We will not go into a detailed analysis here as our primary goal was to demonstrate potential applications of the constructions of preceding sections.

Our protocol has some similarity to the one given in \cite{Hillery}. The difference is ours is probabilistic. However, our protocol accommodates a larger number of players at the expense of requiring multiple copies of states. Thus using 4 qubits ($N=16$) we can accommodate $N^2-N=240$ players. 

The state $\rho$ may also be prepared by a {\em purification} process \cite{NC}. Thus, given a density matrix $\rho$ in an $N$ dimensional Hilbert space $H_N$ we can find a pure state $\ket{\Psi}$ in $H_M\tensor H_N$, where $H_M$ is an $M$-dimensional space, such that 
\[\rho= \tr_{H_M}(\pj{\Psi})\]
The operation $\tr_{H_M}$ is the partial trace with respect to $H_M$. This method eliminates the need for creating a mixed state by random selection. The dealer prepares the state $\ket{\Psi}$ and the players get only the components (qubits) in $H_N$. An added advantage of this procedure is that the dealer can use the fact that $\ket{\Psi}$ is entangled with states in $H_M$ to ensure that the players are using their  assigned measurements. 

The bases given in the equations \eqref{eq:basisRe} and \eqref{eq:basisIm} could be used for quantum state tomography. In this case, we do not require maximal outcomes in all the bases. Thus, assuming the diagonal elements of the density matrix have been estimated we need only measurements with three alternatives. For example, if we want to estimate $\Re(\rho_{12})$ then using the notation in \eqref{eq:basisRe} the orthogonal projections $P^{12}_1$, $P^{12}_2$ and $M^{12}= \sum_{k\neq 1,2} P^{12}_k$ form a complete set and provide the three outcomes. Note also that projections $P^{ij}_i$ and $P^{ij}_j$ can be obtained by repeated applications of Hadamard and control gates. 

\section{Discussion}\label{sec:concl}
In this work we stated and solved a problem of consistent probability assignments for maximal projective measurements, that is, the number of projection operators in each measurement equals the dimension. This can be relaxed. For example, in the measurement bases in Section \ref{sec:dim3} it is immaterial for the basis $\cali{B}_i,\; i>0$ whether we take the complete basis or three projectors---two orthogonal projectors on a ``plane'' spanned by two vectors from $\cali{B}_0$ and one projector orthogonal to the plane. The problem of consistency number $r_n$ arises in the firs place because we have no a priori knowledge about the independence of the projection operators corresponding to the bases. We can visualize such a situation when different observes have no initial communication and can perform only local measurements. Further, it is not easy to define general procedures for constructing such independent bases. One notable exception is the explicit recipe for {\em mutually unbiased bases} (MUB). MUBs were first introduced by Schwinger \cite{Schwinger} in low dimension and later extended to prime and prime power dimensions by Ivanovic \cite{Ivanovic} and Wootters and Fields \cite{Wootters} respectively. The problem is this recipe does not work when the dimension has two distinct prime factors. 

The next logical question would be to consider 
the problem of consistent probability assignment for more general measurement schemes, in particular, for local measurements of entangled states. The most general problem regarding consistency number would be the following. What is the consistency number of a collection of measurements (positive operator valued in general) given some prior information about the state? In this format we have to satisfy some extra constraints. For example, we may have information that the unknown state is {\em pure}. The constraints would be nonlinear in general. More interesting protocols would result from these investigations. We hope to investigate these questions in future.

\end{document}